\documentclass[conference]{IEEEtran}
\IEEEoverridecommandlockouts
\usepackage{amsmath,amssymb,amsfonts, amsthm}
\usepackage{textcomp}
\def\BibTeX{{\rm B\kern-.05em{\sc i\kern-.025em b}\kern-.08em
		T\kern-.1667em\lower.7ex\hbox{E}\kern-.125emX}}
\DeclareMathOperator*{\argmax}{arg\,max}

\DeclareMathOperator{\E}{\mathbb{E}}

\usepackage{graphicx}
\usepackage{cite}
\usepackage{color}
\usepackage{bbm}
\usepackage{textcomp}
\usepackage{multirow}
\usepackage[ruled,vlined]{algorithm2e}
\newcommand{\subparagraph}{}
\usepackage[justification=centering]{caption}
\usepackage[font=footnotesize]{caption}
\usepackage[outdir=./]{epstopdf}
\usepackage{subfigure}
\usepackage{float}
\usepackage[explicit]{titlesec}

\usepackage[letterpaper, left=0.6in, right=0.6in, bottom=0.95in, top=0.7in]{geometry}
\usepackage[caption=false,font=footnotesize]{subfig}

\usepackage[font=footnotesize,skip=2pt]{caption}

\def\BibTeX{{\rm B\kern-.05em{\sc i\kern-.025em b}\kern-.08em
    T\kern-.1667em\lower.7ex\hbox{E}\kern-.125emX}}

\titlespacing{\section}{0pt}{0.01ex}{0.2ex}
\titlespacing{\subsection}{0pt}{0.01ex}{0.1ex}
\titlespacing{\subsubsection}{0pt}{0.01ex}{0.1ex}

\color{black}

\newtheorem{prop}{Proposition}
\newtheorem{definition}{Definition}

\makeatletter

\def\endthebibliography{%
	\def\@noitemerr{\@latex@warning{Empty `thebibliography' environment}}%
	\endlist
}
\makeatother

\let\OLDthebibliography\thebibliography
\renewcommand\thebibliography[1]{
	\OLDthebibliography{#1}
	\setlength{\parskip}{0pt}
	\setlength{\itemsep}{-1pt plus 0.0ex}
}

\begin{document}
\title{\vspace{-.5cm}Efficient Online Learning for Cognitive Radar-Cellular Coexistence via Contextual Thompson Sampling\vspace{-.3cm}
{\footnotesize \textsuperscript{}}
\thanks{To appear in Proc. IEEE Globecom, Taipei, Taiwan, Dec. 2020. $^\dag$C.E. Thornton and R.M. Buehrer are with Wireless@VT, Bradley Department of ECE, Virginia Tech, Blacksburg, VA, 24061. (\textit{Emails: $\{$thorntonc, buehrer$\}$@vt.edu}). $^\ddagger$A.F Martone is with the U.S Army Research Laboratory, Adelphi, MD 20783. (\textit{Email: anthony.f.martone.civ@mail.mil}). The support of the U.S Army Research Office (ARO) is gratefully acknowledged.}}
\author{\IEEEauthorblockN{Charles E. Thornton$^{\dag}$, R. Michael Buehrer$^{\dag}$, and Anthony F. Martone$^{\ddagger}$}}

\IEEEaftertitletext{\vspace{0.001\baselineskip}}
\vspace{-2cm}
\maketitle
\vspace{-3.2cm}
\begin{abstract}	
This paper describes a sequential, or \textit{online}, learning scheme for adaptive radar transmissions that facilitate spectrum sharing with a non-cooperative cellular network. First, the interference channel between the radar and a spatially distant cellular network is modeled. Then, a linear Contextual Bandit (CB) learning framework is applied to drive the radar's behavior. The fundamental trade-off between exploration and exploitation is balanced by a proposed Thompson Sampling (TS) algorithm, a pseudo-Bayesian approach which selects waveform parameters based on the posterior probability that a specific waveform is optimal, given discounted channel information as context. It is shown that the contextual TS approach converges more rapidly to behavior that minimizes mutual interference and maximizes spectrum utilization than comparable contextual bandit algorithms. Additionally, we show that the TS learning scheme results in a favorable SINR distribution compared to other online learning algorithms. Finally, the proposed TS algorithm is compared to a deep reinforcement learning model. We show that the TS algorithm maintains competitive performance with a more complex Deep Q-Network (DQN).
\end{abstract}

\vspace{.4cm}
\section{Introduction}
With the dawn of fifth-generation (5G) cellular technology, the coming years are expected to bring an unprecedented demand for radio frequency spectrum utilization in the 1-6GHz bands. As a result, governing bodies such as the Federal Communications Commission (FCC) and Third Generation Partnership Project (3GPP) have become heavily invested in establishing intelligent secondary access strategies in both licensed and unlicensed frequency bands \cite{factsheet}, \cite{release}. Since statically allocated radars are the largest incumbent consumers of bandwidth in the sub 6GHz bands, practical and robust strategies are necessary to guarantee that both radar and cellular communication systems can meet increasingly stringent performance demands in coexistence scenarios \cite{technical}.

\textit{Related Work.} In the recent literature on coexistence, many contributions have aimed to mitigate mutual interference from the perspective of a communications system through precoding, optimization of transmit waveforms, or estimation of interference channel state information \cite{mlcomms}. Additionally, an array of opportunistic \textit{Cognitive Radar} strategies have been proposed to enhance the performance and interoperability of future radar systems \cite{ontheroad}. Cognitive radar techniques have been recently extended to the domain of spectrum sharing \cite{sharing}, and can be both practical and effective given that a radar's transmitter and receiver are often co-located, allowing channel state information (CSI) to be obtained in real-time via spectrum sensing techniques.

In this light, reinforcement learning (RL) approaches have been proposed to enable cognitive radar systems to optimize transmission parameters given a history of spectrum observations and radar returns \cite{selvi,rlcomp}. However, the application of RL to high-dimensional problems encountered in the real-world such as spectrum sharing often requires a large amount of offline exploration, which be impractical in time-sensitive applications such as radar tracking. Additionally, the dimensionality of traditional RL approaches such as dynamic programming or \textit{Q}-learning can quickly become intractable as the size of the state-action space that defines the problem increases \cite{sutton}. Further, while RL techniques often perform well in cases where the environment obeys the Markov property \cite{selvi}, interference in wireless networks is often both stochastic and dynamic. Thus, better coexistence performance may be achieved by considering extended temporal statistics.  

To mitigate the complexity of solving for an optimal policy in the full RL problem, multi-armed bandit (MAB) approaches are often used for their simplicity and theoretical guarantees \cite{slivkins}. MAB approaches have shown great promise in developing a variety of spectrum access strategies \cite{dsa1},\cite{dsa2}. However, for a time-varying coexistence environment, it is necessary for each system to consider CSI to optimize transmission parameters. Here, we describe a linear contextual bandit (CB) formulation which generalizes the MAB framework by utilizing side information, or \textit{contexts}, derived from the history of transmissions to guide a cognitive radar's decision making such that coexistence with a fixed cellular network is fostered. Thompson sampling (TS), a heuristic for balancing exploration and exploitation in online decision problems used for its favorable practical and theoretical performance \cite{agrawal},\cite{infoTS}. TS is computationally efficient as posterior sampling and distribution parameter updates can be performed efficiently. Thus, TS a natural candidate for RL-driven radar spectrum sharing.

\textit{Contributions.} This work proposes a novel algorithm for radar spectrum sharing based on Thompson sampling that selects radar waveforms based on discounted CSI over extended time scales. The proposed algorithm is practical, and only limited in how fast samples can be drawn from the estimated posterior distribution. We demonstrate that the TS algorithm achieves lower regret in terms of an objective function based on mutual interference and bandwidth utilization than comparable Upper Confidence Bound (UCB) and $\epsilon$-greedy algorithms due to the increased speed of convergence. Further, the linear CB model requires significantly less exploration to learn desirable behavior than a more complex deep reinforcement learning (Deep RL) model, which is advantageous in radar applications where time-sensitive performance is critical, such as target tracking.\\
\section{Coexistence Model}
Consider non-cooperative coexistence between a frequency-agile cognitive radar and $N$ cellular base stations (BSs). The systems must share a $100$ MHz channel centered around $f_{c} = 3.5$ GHz. The channel is divided into $S$ equally-sized sub-bands. Time is slotted. Each time step, or pulse repetition interval, the radar transmits a Linear Frequency Modulated (LFM) chirp waveform, which may occupy any contiguous group of sub-bands within the channel. We assume the radar is located far enough from the cellular network such that small-scale fading effects are negligible. Thus, the interference power received at the radar from BS $i$ is dependent only on path-loss and large-scale shadow fading, and can be written as
\begin{equation}
I_{i} = P_{i}\left\lVert \mathbf{d}_{i} \right\rVert^{- \alpha} \exp({X_{i}}),      \hspace{.6cm} i = 1, ..., N,
\end{equation}
where $P_{i}$ is the BS's transmission power, $\lVert \mathbf{d}_{i} \rVert$ is the distance to the radar, $\alpha$ is the path loss exponent, and $\exp(X_{i})$ is a log-normal random variable to account for shadowing from large obstacles. The aggregate interference $I_{\texttt{agg}} = \sum_{i} I_{i}$ can thus be expressed as a sum of correlated log-normal random variables. This model has been widely studied in wireless communications and is shown to converge to a lognormal limit distribution as the number of BSs becomes large \cite{sum,hume}. We assume that each $X_{i} \sim N(\mu_{i}, \sigma^{2}_{i})$ has a correlation coefficient with $X_{j}$ given by
\begin{equation}
\rho_{ij} = \frac{\E[(X_{i}-\mu_{i})(X_{j}-\mu_{j})]}{\sigma_{i} \sigma_{j}} = \zeta_{i} \zeta_{j} \in [-1, 1].
\end{equation}

Then, with probability 1, the limit distribution as $N \rightarrow \infty$ is
\begin{equation}
I_{\texttt{agg}} \overset{lim}{\sim} \ln N(\mu_{\texttt{agg}}, \lambda^{2}_{\texttt{agg}})
\end{equation}
with parameters 
\begin{equation*}
\begin{split}
\mu_{\texttt{agg}} &= \frac{1}{N} \sum_{i=1}^{N}P_{i} \lVert \mathbf{d}_{i} \rVert^{-\alpha} \exp(\mu_{i}+ \frac{1}{2}\sigma^{2}_{i}),\\
\sigma^{2}_{\texttt{agg}} &= \mu^{2}_{\texttt{agg}}(\exp\left[ (\sigma_{i} \sigma_{j} \zeta_{i} \zeta_{j})^{2} \right] -1 ).
\end{split}
\end{equation*}

We can then characterize the probability that the radar sees harmful interference above power $\mathcal{T}$ as
\begin{equation*}
\mathbb{P}(I_{\texttt{agg}} > \mathcal{T}) = 1 - F_{I_{\texttt{agg}}}(\mathcal{T})
\end{equation*}
\begin{equation}
\label{eq:icdf}
= 1 - \left[ \frac{1}{2}+\frac{1}{2} \operatorname{erf} \left[\frac{\ln{x} - \mu_{\texttt{agg}}}{\sqrt{2} \sigma_{\texttt{agg}}} \right] \right],
\end{equation}
where $F_{I_{\texttt{agg}}}$ is the CDF of $I_{\texttt{agg}}$ and the error function is given by $\operatorname{erf}(p) = \frac{2}{\pi} \int_{0}^{p}e^{-t^{2}} dt$. 

In addition to spatial correlation between nodes $\rho_{ij}$, the cellular interference has temporal correlations due to the block fading structure of the channel. We assume a coherence time of $T_{c}$, meaning the interference is sampled randomly from the distribution of $I_{\texttt{agg}}$ every $T_{c}$ time-steps and remains stationary between samples. The temporal dependencies introduced by this block fading model introduces structure for the frequency-agile radar to learn. A further modeling consideration is the effect of ALOHA-like medium access control. At each time step, each node decides whether to transmit or remain idle with probability $p$ independently of the previous time steps. We denote the subset of transmitting BSs as $N_{\texttt{act}} \subset N$.

The radar must learn an effective transmission strategy to coexist with this interference channel under general conditions. Since this work considers a slow fading model, the radar must collect channel statistics due to the stochastic nature of the interference, but can also expect some temporal correlation based on the coherence time of the channel. We now proceed to a discussion of the learning framework.
\begin{figure}[t]
	\centering
	\includegraphics[scale=0.38]{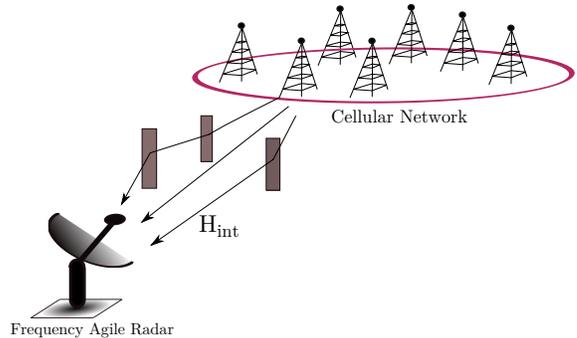}
	\caption{Diagram of the radar-cellular coexistence scenario. Aggregate interference from a cellular network is modeled as the sum of correlated log-normal random variables. }
	\label{fig:model}
\end{figure}
 \\
\section{Contextual Bandit Formulation and Thompson Sampling Algorithm}
This section describes the CB model formulation and the proposed TS algorithm used to drive the cognitive radar's transmission strategy. The radar's action space consists of $K$ arms, $\mathcal{A} = \{ \mathbf{a}_{1}, ..., \mathbf{a}_{K} \}$, where each $\mathbf{a}_{i}$ corresponds to a vector of \emph{contiguous} sub-channels. At each time step $t$, the radar senses the spectral occupancy of the interference $\mathbf{s}_{\texttt{com}}$ and assembles $d$-dimensional context vector for each arm $i$, given by 
\begin{equation}
\mathbf{x}_{i}(t) = [\beta_{1}, \beta_{2}, ..., \beta_{d}],
\end{equation} 
where each $\beta_{j} \in \mathbb{R}$ is a feature determined from the history of actions, contexts, and rewards
\begin{equation}
\begin{split}
\mathcal{H}_{t-1} = \{ \mathbf{a}(\tau), \mathbf{x}_{i}(\tau), r_{\mathbf{a}(\tau)}(\tau) \} \\ i = 1, ..., K,  \hspace{.4cm} \tau = 1, ..., t-1,
\end{split}
\end{equation}
The radar's goal is to maximize a reward signal $r_{i}(t) \in [0,1]$. We assume the reward has mean $\E[r_{i}(t)] = \mathbf{x}_{i}(t)^{T} \boldsymbol{\theta}_{t}$, where $\boldsymbol{\theta}_{t} \in \mathbb{R}^{d}$ is a potentially time-varying parameter vector the radar must learn to achieve optimal average rewards. In other words, at each time step, we assume there exists a vector $\boldsymbol{\theta}_{t}$ that maps the current context to a mean reward for each arm. Since the context is a collection of features associated with an action and the reward is received based on the radar's spectral occupancy, this assumption is reasonable. $r_{i}(t)$, or the reward for arm $i$ at time $t$, is given by

\begin{equation}
r_{i}(t) =
\left \{
\begin{array}{ll}
0, & N_{c} > 0 \\
\eta_{1}/(\eta_{2}N_{mo}), & N_{c} = 0, \hspace{0.1cm} N_{mo} > 0\\
1, &  N_{c} = 0, \hspace{0.1cm} N_{mo} = 0
\end{array} 
\right \},
\label{eq:rwd}
\end{equation}
\noindent where $\eta_{1}$ and $\eta_{2}$ are tunable parameters. $N_{c} \in \{0,1,...,S\}$ and $N_{mo} \in \{0,1,...,S-1\}$ are the number of \textit{collisions} and \textit{missed opportunities} the radar experiences at time $t$. 

\begin{definition}		
	Let the number of collisions $N_{c}$ correspond to the number of sub-channels utilized by both the radar and communication system, $N_{c} = \sum_{i = 1}^{N} \mathbbm{1} \{ \mathbf{a}_{t,i} = \boldsymbol{s}_{\texttt{com}} \}$, where $\mathbbm{1}\{\cdot\}$ is the binary indicator function and the notation $\mathbf{a}_{t,i}$ corresponds to the $i^{th}$ element of the action taken at time $t$. 
\end{definition}

\begin{definition}
	Let the number of missed opportunities $N_{mo}$ correspond to the difference between the largest group of contiguous available sub-channels the radar could possibly occupy, $\mathbf{a}^{*}_{t}$, and the number of those sub-channels the radar actually selects, $N_{mo} = \lVert \mathbf{a}^{*}_{t} \rVert - \sum_{i = 1}^{N} \mathbbm{1} \{ \mathbf{a}^{*}_{t,i} = \boldsymbol{a}_{t,i} \}$.
\end{definition}

The radar's preference for optimal bandwidth utilization versus interference-free transmission is then $\propto \eta_{1}/\eta_{2}$. For large $\eta_{1}/\eta_{2}$, we expect the radar to avoid potential interference channels, as the penalty for missed opportunities is relatively small. For large $\eta_{1}/\eta_{2}$ we expect the radar to act more `aggressively' by attempting to use the entire open bandwidth more often due to the increased penalty for missed opportunities. In a coexistence setting we seek an $\eta_{1}/\eta_{2}$ value which balances this trade-off. The radar requires enough bandwidth to achieve sufficient range resolution for target tracking, while the communication system requires bandwidth to meet throughput requirements. Sufficient SINR, or few collisions, is required for both systems to maintain energy efficiency.

A benefit of the CB framework opposed to other RL formulations is that the action space can be easily extended to a continuous variable \cite{slivkins}.
\begin{prop}
	The reward function given in (\ref{eq:rwd}) is Lipschitz continuous, i.e for two arms $\mathbf{a}_{j}$ and $\mathbf{a}_{k}$, $|\E[r_{\mathbf{a}_{j}}] - \E[r_{\mathbf{a}_{k}}]| \leq \mathcal{D}(j,k)$ where $\mathcal{D}$ is a metric of distance between arms known to the algorithm. Thus, adaptive discritization can be used to extend this formulation to a continuous set of actions with upper and lower bounded regret in a stationary setting.
\end{prop}
\begin{proof} \renewcommand{\qedsymbol}{}
	Let $\mathcal{D}(j,k)$ be the Hamming distance between the vectors of binary values representing $\mathbf{a}_{j}$ and $\mathbf{a}_{k}$. Since $r_{i}(t)$ is calculated based on missed opportunities and collisions, if $j$ and $k$ have a small Hamming distance then $\E[r_{\mathbf{a}_j}] \approx \E[r_{\mathbf{a}_{k}}]$ due to spectral overlap.
\end{proof}
\vspace{-2mm}
By learning to maximize rewards, the radar is equivalently aiming to minimize the total \textit{regret} experienced before time $T$. Regret corresponds to the total difference between the agent's received reward and the reward obtained by taking the best action at each time step given by
\begin{equation}
\label{eq:regret}
\mathrm{Regret} = \textstyle \sum_{t = 1}^{T} r_{\mathbf{a}^{*}(t)}(t) - r_{\mathbf{a}(t)}(t),
\end{equation}
where $\mathbf{a}^{*}(t)$ is the action with the highest expected reward at time $t$ and $\mathbf{a}(t)$ is the action actually selected by the radar at time $t$. In the case of linear payoffs, (\ref{eq:regret}) is equivalent to
\begin{equation*}
\mathbf{x}_{\mathbf{a}^{*}(t)}(t)^{T} \boldsymbol{\theta}_{t} - \mathbf{x}_{i}^{T} \boldsymbol{\theta}_{t}.
\vspace{-.1mm}
\end{equation*}

As the true value of $\boldsymbol{\theta}_{t}$ is not always known to the algorithm, we can retroactively compute the regret at the next time step using $\mathbf{a}^{*}(t-1)$. However, in contrast to some bandit formulations which consider the case of \textit{full-feedback}, or knowledge of the rewards associated with each action at the next time step, the cognitive radar system discussed here only receives feedback based on the action taken.

\begin{algorithm}[t]
	\setlength{\textfloatsep}{0pt}
	\label{algo:ts}
	\caption{Adaptive Radar Thompson Sampling}
	\SetAlgoLined
	Initialize parameters $\mathbf{B} = \mathbf{I}_{d}$, $\hat{\boldsymbol{\theta}}_{t} = \mathbf{0}_{d}$, $\mathbf{f} = \mathbf{0}_{d}$\\
	\For{t = 2, ..., T}{
		Sense cellular interference $\mathbf{s}_{\texttt{com}}$;\vspace{0.1cm}
		
		Using $\mathbf{s}_{\texttt{com}}$ and $\mathcal{H}_{t-1}$ assemble \textit{discounted} context $\mathbf{x}_{i}(t) = \{\beta_{1},...,\beta_{d}\} \; \; \forall \; i$;\vspace{0.1cm}
		
		Sample $\tilde{\boldsymbol{\theta}}_{t} \sim \mathcal{N}(\hat{\boldsymbol{\theta}}_{t}, v^{2} \mathbf{B}^{-1})$;\vspace{0.1cm}
		
		Create constrained action space $\mathcal{A'}_{t} = \{ \mathbf{a} \in \mathcal{A}: \E[r_{i}(t)|\mathbf{x}_{t} = x] > \hat{r} \}$; \vspace{0.1cm}
		
		Select LFM waveform $\mathbf{a}(t) = \argmax_{i} \mathbf{x}_{i}(t)^{T} \tilde{\boldsymbol{\theta}}_{t}$;\vspace{0.1cm}
		
		Observe reward $r_{i}(t)$ from (\ref{eq:rwd});\vspace{0.1cm}
		
		Update history $\mathcal{H}_{t}$;\vspace{0.1cm}
		
		Update parameters $\mathbf{B} = \mathbf{B} + \mathbf{x}_{\mathbf{a}(t)} \mathbf{x}_{\mathbf{a}(t)}^{T}$, $\mathbf{f} = \mathbf{f} + \mathbf{x}_{\mathbf{a}(t)}(t)r_{\mathbf{a}(t)}(t)$, and $\hat{\mu}=\mathbf{B}^{-1}f$;
	}
\end{algorithm}	

While the interference channel model here follows a log-normal limiting distribution, we seek to optimize performance in a general setting. The TS algorithm discussed here uses the normal-normal conjugacy property to formulate a posterior distribution, from which we sample estimate $\tilde{\boldsymbol{\theta}}_{t}$. We now proceed to a description of the algorithm, which is also seen in Algorithm 1.

Given context $\mathbf{x}_{i}(t)$ and parameter $\boldsymbol{\theta}_{t}$, the likelihood of reward $r_{i}(t)$ is
\begin{equation}
\mathcal{L}(r_{i}(t) | \boldsymbol{\theta}_{t}) \sim \mathcal{N}(\mathbf{x}_{i}(t)^{T} \boldsymbol{\theta}_{t}, v^{2}).
\end{equation}
We can then place a Gaussian prior distribution on $\boldsymbol{\theta}$ given by
\begin{equation}
\mathbb{P}(\boldsymbol{\theta}_{t}) \sim \mathcal{N} (\hat{\boldsymbol{\theta}}_{t}, v^{2} \mathbf{B}(t)^{-1}),
\end{equation}
where $v$ is an exploration parameter that specifies the algorithm. Applying Bayes' rule, the posterior distribution on $\boldsymbol{\theta}_{t}$ can then be written up to proportionality as
\begin{equation}
\begin{split}
\mathbb{P}(\tilde{\boldsymbol{\theta}_{t}} | r_{i}(t)) & \propto \mathcal{L}(r_{i}(t) | \boldsymbol{\theta}_{t}) \mathbb{P}(\boldsymbol{\theta}_{t}) \\
& \propto \mathcal{N}(\hat{\boldsymbol{\theta}}_{t}, v^{2}\mathbf{B}(t)^{-1}),
\end{split}
\end{equation}
\noindent where the posterior mean and covariance matrix can be expressed as
\begin{equation*}
\begin{split}
&\mathbf{B}(t) = \mathbf{I}_{d} + \textstyle \sum_{\tau = 1}^{t-1} \mathbf{x}_{\mathbf{a}(\tau)}(\tau) \mathbf{x}_{\mathbf{a}(\tau)}^{T}(\tau), \\
&\hat{\boldsymbol{\theta}}_{t} = \mathbf{B}(t)^{-1} \textstyle \sum_{\tau=1}^{t-1} \mathbf{x}_{\mathbf{a}(\tau)}(\tau) r_{\mathbf{a}(\tau)}(\tau),
\end{split}
\end{equation*}
\noindent where $\mathbf{I}_{d}$ is the $d$-dimensional identity matrix. Thus, the posterior estimate $\tilde{\boldsymbol{\theta}}_{t}$ can be efficiently sampled from a $d$-dimensional multivariate normal distribution. The distribution parameters $\hat{\boldsymbol{\theta}}_{t}$ and $\mathbf{B}^{-1}$ can also be easily updated based on the context and reward received at time $t$. To update $\mathbf{B}^{-1}$ without taking a matrix inversion at every step, we can apply the Sherman-Morrison formula, which allows the update to be efficiently computed by
\begin{equation}
(\mathbf{B}+ \mathbf{x}_{\mathbf{a}(t)} \mathbf{x}^{T}_{\mathbf{a}(t)})^{-1} = \mathbf{B}^{-1} + \frac{\mathbf{B}^{-1}\mathbf{x}_{\mathbf{a}(t)} \mathbf{x}^{T}_{\mathbf{a}(t)} \mathbf{B}^{-1}}{1 + \mathbf{x}^{T}_{\mathbf{a}(t)} \mathbf{B}^{-1} \mathbf{x}_{\mathbf{a}(t)}}.
\end{equation} 

Based on the analysis in \cite{infoTS}, we know that the radar's estimated model $\boldsymbol{\theta}_{t}$ incurs some regret whenever the Kullback-Leibler (KL) divergence between estimated $\tilde{\boldsymbol{\theta}}_{t}$ and true model $\boldsymbol{\theta}^{*}_{t}$, $D(\boldsymbol{\theta}^{*}_{a,t} || \tilde{\boldsymbol{\theta}}_{a,t}) > 0$. Thus, a TS algorithm will work well in practice whenever the estimated model satisfies our assumption of linear payoffs, which implies that the spectrum feature set must accurately describe the mapping between context-action pairs and rewards. However, $\boldsymbol{\theta}_{t}$ in general evolves according to a time-varying stochastic process and becomes difficult to predict as the entropy of the distribution on the optimal action $H(\mathbb{P}(\mathbf{a}^{*}_{t}) )$ grows. Thus, the cognitive radar's performance will be limited by the temporal correlations in the observed interference.

\begin{figure*}[t]
	\vspace{-4mm}
	\centering
	\includegraphics[scale=0.47]{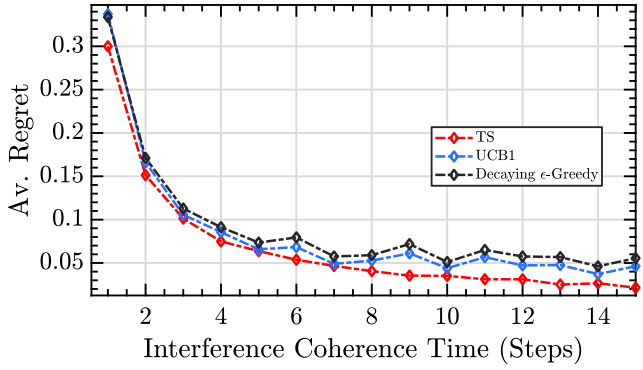}
	\hspace{3mm}
	\includegraphics[scale=0.48]{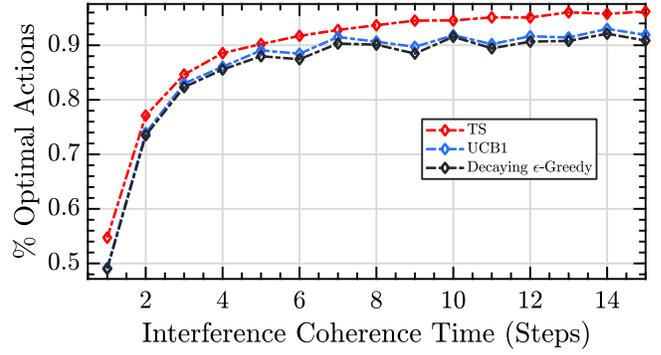}
	\caption{LEFT: Average regret per time step of online learning algorithms with varying channel coherence time. RIGHT: Probability each algorithm selects the optimal action over with channel coherence time. Each data point corresponds to $10^{4}$ simulated time steps.}
	\label{fig:optComp}
	\vspace{-7mm}
\end{figure*}

Since $r_{i}(t)$ is a linear function of $N_{c}$ and $N_{mo}$, the radar can reasonably estimate $\E[r_{i}(t)]$ from the history of collisions, missed opportunities, and rewards for a context-action pair or an action across all contexts. Further, the radar can consider either recent feedback or information averaged across many time steps depending on the specification of $\mathbf{x}_{i}(t)$. Here, we consider context which incorporates the following features
\begin{equation}
\begin{split}
\beta_{1} &= \frac{1}{n_{p}} \textstyle \sum_{k = 1}^{n_{p}} r_{i}(k), \\
\beta_{2} &= \frac{1}{n_{p} - 1} \textstyle \sum_{k = 1}^{n_{p}} (r_{i}(k) - \E[r_{i}])^{2},\\
\beta_{3} &= r_{i}(n_{p}),
\end{split}
\end{equation}
where $n_{p}$ is the number of times arm $i$ has been played and $k$ is an index of the history. This context formulation weighs the average reward received, variance of rewards, and the previous reward to account for changing channel conditions. 

In coexistence settings, the distribution of the interference channel can vary over time and is in general non-stationary, which in the bandit literature is known as a \textit{restless bandit} problem \cite{slivkins}. In this framework, we model uncertainty about the channel by considering parameter vector $\boldsymbol{\theta}_{t}$ that evolves over time. Uncertainty about the channel must also be reflected in the posterior distribution, which can be modeled by updating the posterior distribution parameters to weigh recent information more heavily than observations from the distant past, a process known as \textit{discounting}.

Here, we inject uncertainty by ignoring observations more distant than $\tau$ time steps in the past. This prevents the posterior mass from concentrating heavily around one value due to the limited number of observations. Additionally, when assembling context vectors, we weigh observations by factor $\gamma^{k}$, where $k$ is the number of time steps elapsed since the observation. As the posterior covariance matrix is updated based on $\mathbf{x}_{\mathbf{a}(t)} \mathbf{x}_{\mathbf{a}(t)}^{T}$, discounted observations reflect an increase in uncertainty as the channel dynamics change.

Another modeling consideration for the frequency-agile radar system is \textit{caution} about potentially hazardous actions. We assume that both the radar and cellular system wish to maintain some minimum outage probability $\mathbb{P}(\texttt{SINR} > \mathcal{T})$. Thus at each round $t$, we consider the constrained action space
\begin{equation}
\mathcal{A'}_{t} = \{ \mathbf{a} \in A: \E[r_{\mathbf{a}}(t)|\mathbf{x}_{t} = x] > \hat{r} \},
\end{equation}
where $\hat{r}$ is a reward value such that $\E[\texttt{SINR}] > \mathcal{T}$. To estimate $\texttt{SINR}$ at the radar, we use
\begin{equation}
\texttt{SINR} = \frac{P_{r}}{P_{I}+P_{N}} = \frac{P_{r}e^{-\psi}}{P_{N}+ \sum_{i = 1}^{N_{c}} P_{i} d_{i}^{-\alpha} e^{-x_{i}}},
\end{equation}
where $\psi \sim N(\mu_{\psi},\sigma^{2}_{\psi})$ is a random variable to account for fluctuations in received power due to the target model. While no closed form distribution exists for the radar's $\texttt{SINR}$ and outage probability, they can be lower bounded, as in \cite{bound}.

Now that the linear CB model and proposed TS algorithm have been described, we proceed to compare the learning framework to other online learning schemes as well as a Deep RL approach based on Deep \textit{Q}-learning.\\

\section{Simulation Study}
To validate the utility of the proposed coexistence model and TS algorithm, several simulation comparisons are presented. Here, the 100MHz channel is divided into $S = 10$ equally sized sub-channels. The communication system bandwidth is $20$MHz, and interference occupies the second and third sub-channels when present. There are a total of $N = 120$ randomly scattered cellular BS's, each located $4$-$6$ Km from the radar. The path loss exponent $\alpha=4$ characterizes the shadowing environment. The base station transmission power $P_{i}$ ranges from $40$-$46.5$ dBm, consistent with the upper range of the 3GPP standard. The reward parameters are $\eta_{1} = 10$ and $\eta_{2} = 11$, which effectively balances missed opportunities and collisions for the case of $S = 10$ sub-channels.

\subsection{Comparison to Other Online Learning Algorithms}

The proposed TS algorithm is compared to a decaying $\epsilon$-greedy algorithm, which selects a random action with probability $1- \epsilon_{t}$ and action 
\begin{equation}
\mathbf{a}^{*} = \argmax_{\mathbf{a} \in \mathcal{A}} \{ \E{[r_{t+1}|\mathbf{x}(t), \mathcal{H}_{t-1}, \mathbf{a}(t) = \mathbf{a}]} \}
\end{equation}
with probability $\epsilon_{t}$. We choose an initial value $\epsilon_{0} = .95$ and induce a decay of $\epsilon_{t} = \epsilon_{t-1}  - .001$ every time step. Additionally, we compare performance to that of the UCB1 algorithm, based on the well-known \textit{upper confidence bound} family of algorithms which can be thought of as a frequentist companion to TS, as both expected reward and uncertainty are considered. The UCB1 algorithm is given in Algorithm 2.

\NoCaptionOfAlgo		
\begin{algorithm}[h]
	\SetAlgoLined
	Play each arm once;\\
	Then play $\max{[\hat{\xi}_{k}(t)+\xi_{t}(a)]}$, where $\hat{\xi}_{k}(t)$ is the expected reward for arm $k$ at $t$, $\xi(a) = \sqrt{\frac{2 \log (t)}{n_{k}(t)}}$ is the \emph{confidence radius}, and $n_{k}(t)$ is the number of times arm $k$ has been played, until time $T$;	
	\caption{\textbf{Algorithm 2}: UCB1}
\end{algorithm}	

First, we seek to analyze the effects of environmental uncertainty on the performance of the proposed TS algorithm through varying channel coherence time $T_{c}$. For a given value of $T_{c}$, we assume that the cellular interference is constant for at least $T_{c}$ time steps, so we sample the sum of lognormal interferers at integer multiples of $T_{c}$. Based on the information theoretic interpretation of TS, we expect performance to improve smoothly with increasing $T_{c}$ as the best action $\mathbf{a}^{*}(t)$ becomes less random.

In Fig. \ref{fig:optComp}, we see that as the channel coherence time increases, each of the online learning algorithms experiences a lower average regret and selects optimal $\mathbf{a}^{*}(t)$ a higher percentage of the time. The proposed TS algorithm has the highest selection of optimal actions for each value of $T_{c}$ tested and performs particularly well for longer coherence times. The proposed TS algorithm also performs at least as well as UCB1 and decaying $\epsilon$-greedy in terms of average reward, with the gap between TS and the other algorithms becoming larger for longer coherence times.

\begin{figure}[b]
	\vspace{-3mm}
	\centering
	\includegraphics[scale=0.5]{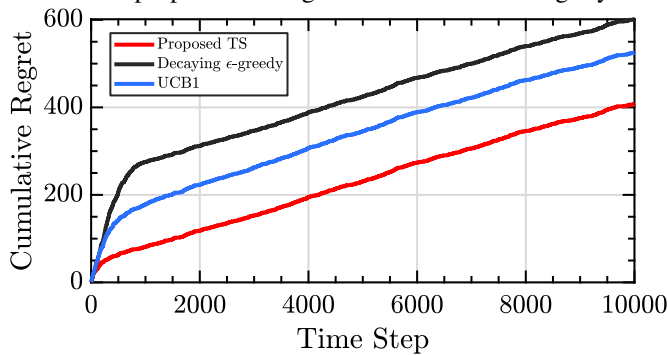}
	\caption{Cumulative regret of online learning algorithms for the case of $T_{c}$ = 8 steps. The proposed TS algorithm demonstrates quicker convergence than UCB1 and decaying $\epsilon$-greedy.}
	\label{fig:convergenceSpeed}
\end{figure}

To gain some insight as to why the TS algorithm appears to perform well, we can look to convergence in terms of cumulative regret over a fixed learning period. In Fig. \ref{fig:convergenceSpeed}, we see that the TS algorithm begins to convergence quicker in terms of regret than the UCB1 and decaying $\epsilon$-greedy algorithms. However, upon convergence all three online learning algorithms perform similarly as they each learn the same linear feature set. Thus we can attribute radar performance difference between the online learning algorithms to rate of convergence.

While these results indicate good performance in terms of the proposed spectrum-sharing oriented reward function, we also seek to optimize radar performance metrics of interest. To further characterize the performance of the proposed TS algorithm, we examine the distribution of $\texttt{SINR}$ at the radar. In Fig. \ref{fig:SINR}, we see the empirical CDF of $\texttt{SINR}$ from $10^{4}$ radar transmissions using a fixed full bandwidth approach in addition to TS, UCB1, and decaying $\epsilon$-greedy. Firstly, we observe a significant improvement from utilizing the adaptive online learning schemes in comparison to traditional fixed bandwidth radar. Further, we notice that in the tail of the distribution, we see that the proposed TS algorithm maintains a slightly lower probability of very low SINR values than UCB1 or $\epsilon$-greedy. This can be attributed to the reduced time spent exploring potentially costly actions. Since very low $\texttt{SINR}$ values are the primary cause of missed detections, we immediately see the value of reduced exploration from a performance standpoint.

\begin{figure}[b]
	\centering	
	\includegraphics[scale=0.47]{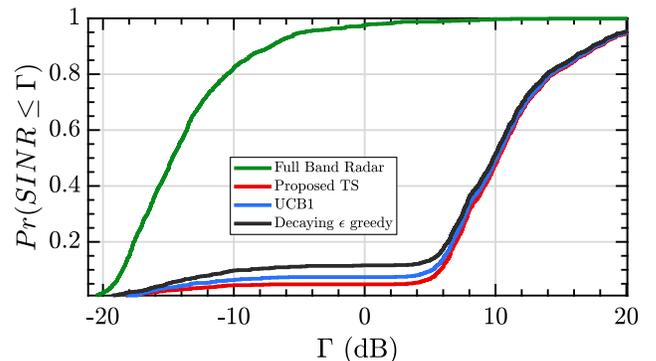}
	\caption{CDF of observed $\texttt{SINR}$ values for full bandwidth radar compared to the proposed TS, UCB1, and decaying $\epsilon$-greedy approaches during $10^{4}$ time steps of radar operation in the case of $T_{c} = 10$ steps.}
	\label{fig:SINR}
\end{figure}

\subsection{Comparison to Deep Reinforcement Learning}

\begin{figure}[t]
	\centering
	\includegraphics[scale=0.43]{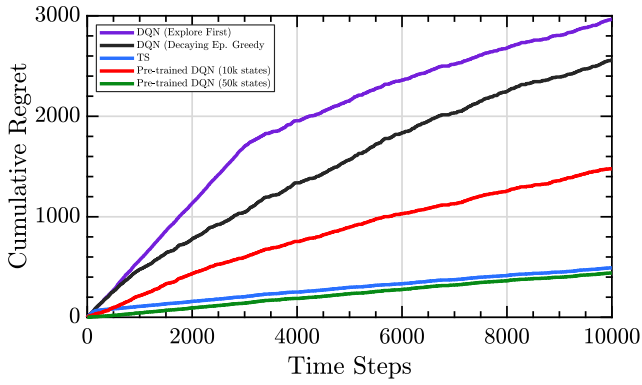}
	\caption{A comparison of the regret accumulated by a DQN with various amounts of training to the proposed TS algorithm.}
	\label{fig:fullCom}
\end{figure}
Previous work \cite{rlcomp}, has proposed Deep RL algorithms to control spectrum sharing radar systems. In RL, the radar's actions are assumed to have an influence on the future states of the interference environment. However, in radar-cellular coexistence scenarios, this may or may not be the case depending on the spectral environment and cellular network configuration. If a function approximation approach is used, such as $Q$-learning, then the decision maker can learn online and adapt to changing environmental conditions.

\begin{figure}[t]
	\vspace{-2mm}
	\centering
	\includegraphics[scale=0.5]{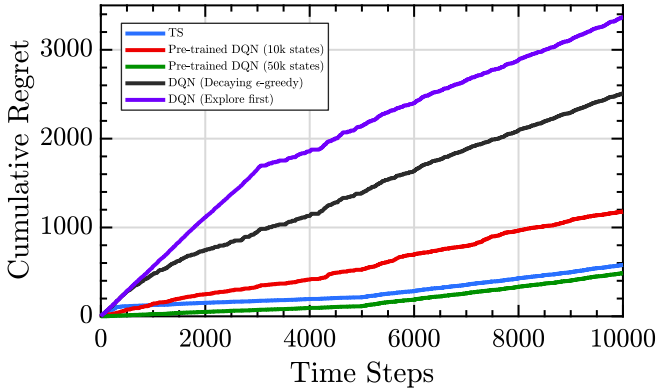}
	\caption{Cumulative regret accumulated by a DQN with various training experiences and the proposed TS algorithm when the channel changes from $T_{c} = 14$ steps to $T_{c} = 4$ steps halfway through a $10^{4}$ step learning period.}
	\label{fig:gen}
\end{figure}
A notable advantage of Deep RL over linear models is that it can be used to approximate nonlinear mappings between environmental states and rewards. However, deep neural networks are often time-consuming and computationally burdensome to train. Additionally, NNs consist of large parameter spaces and often require a great deal of exploration to find a set of weights which leads to good performance. In Fig. \ref{fig:fullCom}, we compare the regret incurred by a 3-layer Deep $Q$-Network (DQN) to the proposed TS algorithm for a lognormal sum interference channel with $T_{c} = 7$ steps. 

We see that when the DQN-enabled radar is allowed 3,000 steps to explore uniformly and then picks the greedy action thereafter, the radar continues to incur a large amount of regret during exploitation. When the DQN explores online with a decaying $\epsilon$-greedy strategy, convergence is smoother than the explore-first strategy. However, upon convergence, performance is still much worse than the proposed TS algorithm. When the network is trained offline for $10^{4}$ steps and takes the greedy action during the entire trial of $10^{4}$ pulses, performance begins to improve drastically. With $5 \times 10^{4}$ steps of offline training, the DQN achieves better performance than the proposed TS algorithm.

In Fig. \ref{fig:gen} we observe the cumulative regret when the interference changes from $T_{c} = 14$ steps to $T_{c} = 4$ steps halfway through the evaluation run of $10^{4}$ pulses. We see that even when the DQN is pre-trained for $5 \times 10^{4}$ PRIs, it adapts to the change in interference complexity no more effectively than the proposed TS algorithm. This result shows that in non-stationary interference scenarios, the contextual TS model is an effective approach.

Thus, while Deep RL presents a powerful nonlinear hypothesis class for learning complex interference patterns with online learning capabilities, a large amount of exploration is often necessary to achieve better performance than a linear contextual bandit model.\\

\section{Conclusions}
This work has presented a contextual Thompson sampling strategy for coexistence between a cognitive radar and a cellular network modeled as a sum of lognormal interference sources. A Thompson Sampling (TS) algorithm was presented to efficiently balance the fundamental trade-off between exploration and exploitation. Due to an increased speed of convergence, the proposed TS algorithm performs better than comparable online learning algorithms in terms of regret calculated from a weighted combination of interference avoidance (collisions) and bandwidth utilization (missed opportunities) as well as in the distribution of observed $\texttt{SINR}$. 

Further, the proposed contextual bandit TS approach provides some key improvements over Deep RL cognitive control. The TS algorithm allows for efficient exploration of the state space and can scale to larger action spaces than the Deep RL approach, resulting in superior convergence time when learning online. Further, we show that while Deep RL is able to achieve better asymptotic performance, a significant offline exploration phase is necessary. 

Given that radar applications often demand rapid reaction to find and reliably track a moving target with minimal mutual interference, efficient online learning is an important consideration for spectrum sharing radar systems. Future work could include modeling a cognitive communications system. Additionally, this could be extended to balance the spectral and energy efficiency in a distributed radar network. \\

\scriptsize
\bibliographystyle{IEEEtran}
\bibliography{bibli}

\begin{thebibliography}{10}
\providecommand{\url}[1]{#1}
\csname url@samestyle\endcsname
\providecommand{\newblock}{\relax}
\providecommand{\bibinfo}[2]{#2}
\providecommand{\BIBentrySTDinterwordspacing}{\spaceskip=0pt\relax}
\providecommand{\BIBentryALTinterwordstretchfactor}{4}
\providecommand{\BIBentryALTinterwordspacing}{\spaceskip=\fontdimen2\font plus
\BIBentryALTinterwordstretchfactor\fontdimen3\font minus
  \fontdimen4\font\relax}
\providecommand{\BIBforeignlanguage}[2]{{%
\expandafter\ifx\csname l@#1\endcsname\relax
\typeout{** WARNING: IEEEtran.bst: No hyphenation pattern has been}%
\typeout{** loaded for the language `#1'. Using the pattern for}%
\typeout{** the default language instead.}%
\else
\language=\csname l@#1\endcsname
\fi
#2}}
\providecommand{\BIBdecl}{\relax}
\BIBdecl

\bibitem{factsheet}
{Federal Communications Commission (FCC)}, ``{Spectrum Horizons},'' \emph{Fed.
  Registrar}, vol.~84, no. 107, Jun. 2019.

\bibitem{release}
{Third Generation Partnership Project (3GPP)}, ``{5G in Release 17 - Strong
  Radio Evolution},'' White Paper.

\bibitem{technical}
H.~Griffiths \emph{et~al.}, ``{Radar Spectrum Engineering and Management:
  Technical and Regulatory Issues},'' \emph{Proc. IEEE}, vol. 103, no.~1, pp.
  85--102, Jan. 2015.

\bibitem{mlcomms}
L.~Zheng \emph{et~al.}, ``{Radar and Communication Coexistence: An Overview: A
  Review of Recent Methods},'' \emph{IEEE Sig. Proc. Mag.}, vol.~36, no.~5, pp.
  85--99, Sep. 2019.

\bibitem{ontheroad}
M.~S. Greco \emph{et~al.}, ``{Cognitive Radars: On the Road to Reality Progress
  Thus Far and Possibilities for the Future},'' \emph{IEEE Sig. Proc. Mag.},
  vol.~35, no.~4, pp. 112--125, Jul. 2018.

\bibitem{sharing}
P.~Stinco, M.~S. Greco, and F.~Gini, ``{Spectrum Sensing and Sharing for
  Cognitive Radars},'' \emph{IET Radar, Sonar, and Nav.}, vol.~10, no.~3, Feb.
  2016.

\bibitem{selvi}
E.~Selvi \emph{et~al.}, ``{Reinforcement Learning for Adaptable Bandwidth
  Tracking Radars},'' \emph{IEEE Trans. Aero. and Elec. Sys.}, To Appear 2020.

\bibitem{rlcomp}
C.~E. Thornton \emph{et~al.}, ``{Experimental Analysis of Reinforcement
  Learning Techniques for Spectrum Sharing Radar},'' \emph{{in Proc. IEEE Intl.
  Radar Conf.}}, Apr. 2020.

\bibitem{sutton}
R.~S. Sutton and A.~G. Barto, \emph{{Reinforcement Learning: An Introduction}},
  2018.

\bibitem{slivkins}
A.~Slivkins, \emph{{Introduction to Multi-Armed Bandits}}.\hskip 1em plus 0.5em
  minus 0.4em\relax Now Publishers, 2019.

\bibitem{dsa1}
P.~M. Navikkumar~Modi and C.~Moy, ``{QoS Driven Channel Selection Algorithm for
  Cognitive Radio Network: Multi-User Multi-Armed Bandit Approach},''
  \emph{IEEE Trans. Cog. Commun. and Netw.}, vol.~3, no.~1, Mar. 2017.

\bibitem{dsa2}
M.~Khaledi and A.~A. Abouzeid, ``{Dynamic Spectrum Sharing Auction With
  Time-Evolving Channel Qualities},'' \emph{IEEE Trans. Wireless Commun.},
  vol.~14, no.~11, Nov. 2015.

\bibitem{agrawal}
S.~Agrawal and N.~Goyal, ``{Thompson Sampling for Contextual Bandits with
  Linear Payoffs},'' \emph{Intl. Conf. on Mach. Lrn. (ICML)}, Jun. 2013.

\bibitem{infoTS}
D.~J. Russo and B.~V. Roy, ``{An Information-Theoretic Analysis of Thompson
  Sampling},'' \emph{J. Mach. Learn. Research}, vol.~17, Apr. 2016.

\bibitem{sum}
N.~Beaulieu and Q.~Xie, ``{An optimal lognormal approximation to lognormal sum
  distributions},'' \emph{IEEE Trans. Veh. Tech.}, vol.~53, no.~2, Mar. 2004.

\bibitem{hume}
A.~Khawar, A.~Abdelhadi, and T.~C. Clancy, ``{A mathematical analysis of
  cellular interference on the performance of S-band military radar systems},''
  \emph{in Proc. IEEE Wireless Telecom. Symp.}, Apr. 2014.

\bibitem{bound}
F.~Berggren and S.~Slimane, ``{A simple bound on the outage probability with
  lognormally distributed interferers},'' \emph{IEEE Commun. Letters}, vol.~8,
  no.~5, pp. 271--273, May 2004.

\end{thebibliography}

\end{document}